\documentclass[%
 aip,
 amsmath,amssymb,
 reprint,%
 onecolumn
]{revtex4-2}

\pdfoutput=1

\usepackage{graphicx}% Include figure files
\usepackage{dcolumn}% Align table columns on decimal point
\usepackage{bm}% bold math
\usepackage[utf8]{inputenc}
\usepackage[T1]{fontenc}
\usepackage{mathptmx}
\usepackage{etoolbox}

\makeatletter
\def\@email#1#2{%
 \endgroup
 \patchcmd{\titleblock@produce}
  {\frontmatter@RRAPformat}
  {\frontmatter@RRAPformat{\produce@RRAP{*#1\href{mailto:#2}{#2}}}\frontmatter@RRAPformat}
  {}{}
}%
\makeatother

\newcommand{\eps}{\varepsilon}
\newcommand{\X}{\mathcal{X}}
\newcommand{\N}{\mathbb{N}}
\newcommand{\Z}{\mathbb{Z}}
\newcommand{\R}{\mathbb{R}}
\newcommand{\C}{\mathbb{C}}
\newcommand{\Hi}{\mathcal{H}}
\newcommand{\abs}[1]{\left\lvert#1\right\rvert}
\newcommand{\FC}{{\mathcal{C}}}
\newcommand{\set}[1]{\left\{#1\right\}}
\newcommand{\eu}{\mathrm{e}}
\newcommand{\iu}{\mathrm{i}}
\newcommand{\di}{\mathrm{d}} 
\newcommand{\sub}[1]{_{\mathrm{#1}}}
\newcommand{\subm}[2]{_{\mathrm{#1},#2 }}%%sub modified
\newcommand{\su}[1]{^{\mathrm{#1}}}
\newcommand{\Id}{\mathbf{1}}
\newcommand{\norm}[1]{\left\| #1 \right\|}
\newcommand{\vertiii}[1]{{\left\vert\kern-0.25ex\left\vert\kern-0.25ex\left\vert #1 \right\vert\kern-0.25ex\right\vert\kern-0.25ex\right\vert}}
\newcommand{\SC}{\mathcal{S}}
\newcommand{\LI}{\mathcal{L}}
\newcommand{\Or}{\mathcal{O}}
\newcommand{\G}{\mathcal{G}}

\DeclareMathOperator{\Tr}{Tr} 
\DeclareMathOperator{\re}{Re}
\DeclareMathOperator{\Ran}{Ran}

\usepackage{amsthm}

\theoremstyle{plain}
\newtheorem{theorem}{Theorem}[section]

\newtheorem{corollary}[theorem]{Corollary}
\newtheorem{proposition}[theorem]{Proposition}

\theoremstyle{definition}

\newtheorem{assumption}[theorem]{Assumption}

\usepackage{enumerate}
\usepackage{enumitem}

\begin{document}

\preprint{arXiv:2203.08044}

\title[From charge to spin transport coefficients]{From charge to spin: analogies and differences in quantum transport coefficients}
% Force line breaks with \\
\author{Giovanna Marcelli}
 \homepage{Email: \href{mailto:giovanna.marcelli@sissa.it}{giovanna.marcelli@sissa.it}}
 \affiliation{Mathematics Area, SISSA\\ Via Bonomea 265, 34136 Trieste, Italy}
\author{Domenico Monaco}%
 \homepage{Email: \href{mailto:monaco@mat.uniroma1.it}{monaco@mat.uniroma1.it}}
 \affiliation{Dipartimento di Matematica, ``Sapienza'' Universit\`{a} di Roma\\ Piazzale Aldo Moro 5, 00185 Rome, Italy
%\\This line break forced with \textbackslash\textbackslash
}%

\date{\today}% It is always \today, today,

\begin{abstract}
We review some recent results from the mathematical theory of transport of charge and spin in gapped crystalline quantum systems. The emphasis will be in transport coefficients like conductivities and conductances. As for the former, those are computed as appropriate expectations of current operators in a \emph{non-equilibrium almost-stationary state} (NEASS), which arises from the perturbation of an equilibrium state by an external electric field. While for charge transport the usual double-commutator Kubo formula is recovered (also beyond linear response), we obtain formulas for appropriately-defined spin conductivities which are still explicit but more involved. Certain ``Kubo-like'' terms in these formulas are also shown to agree with corresponding contributions to the spin conductance. In addition to that, we employ similar techniques to show a new result, namely that even in systems with non-conserved spin there is no generation of spin torque, that is the spin torque operator has an expectation in the NEASS which vanishes faster than any power of the intensity of the perturbing field.
\end{abstract}

\maketitle

\section{Introduction}

This survey article reviews some recent results obtained by the authors and collaborators in the context of quantum transport of charge and spin in crystalline solids. The main focus is on the response of appropriate current observables once the system, initially at equilibrium, is perturbed by the presence of an external electric field (resp.\ voltage difference); of particular interest is the computation of the linear response coefficient in the form of a conductivity (resp.\ conductance). This setup appears for example in the physics of the quantum Hall and quantum spin Hall effects \cite{vonKlitzing80,SinovaReview}, where transverse conductivities have particular relevance due to their topological origin: for this reason, we will mostly focus on two-dimensional crystals and consider the response of a current in a direction transverse to the inducing perturbation.

While the standard approach to the computation of linear response coefficients, which we briefly review in Section~\ref{sec:heuristic}, requires that the perturbing field is modulated slowly (adiabatically) in time, we advocate here for a different method. Indeed, the combination of the time-adiabatic limit and the zero-field limit poses several conceptual and technical difficulties. Therefore, we propose to avoid the time modulation of the perturbing field altogether, and follow ideas inspired by what is known in physics and chemistry as Rayleigh--Schr\"odinger perturbation theory \cite{SakuraiNapolitano}, and which have been made rigorous in the mathematical physics community in the context of space-adiabatic perturbation theory (see Refs.\ \onlinecite{PanatiSpohnTeufel03, Teufel03} and references therein). This alternative method has been advanced by S.~Teufel and collaborators in a series of works \cite{MonacoTeufel17, Teufel20, HenheikTeufelI, HenheikTeufelII} and applied in the present context of quantum transport also in collaboration with the authors \cite{MaPaTe, MarcelliMonaco, MarcelliTeufel}. The approach allows to identify a so-called \emph{non-equilibrium almost stationary state} (NEASS), which in the adiabatic regime well approximates the physical state once the dynamical switching drives the system out of equilibrium, but which can be defined without resorting to a time-dependent modulation of the perturbing field (see Theorem~\ref{thm:constr NEASS} for a precise stament). On the other hand, the construction of the NEASS is fairly explicit, and allows in particular to expand it in powers of the strength $\eps$ of the external field around the unperturbed equilibrium state (which, under a spectral gap assumption, we take to be the Fermi projection onto the occupied energy bands of the reference system). For charge currents, the use of the NEASS allows in particular to recover the well-known Kubo formula (sometimes called Kubo--Chern formula or double-commutator formula, DCF) for the Hall conductivity, and to prove its validity also beyond linear response (see Refs.\ \onlinecite{KleinSeiler,Bachmannetal,MarcelliMonaco} and Theorem~\ref{thm: sigmahall} below).

Having identified the perturbed state of the system out of equilibrium, one can probe the response of other observables, like for example a spin current. While the appropriate definition of the spin current operator is still debated in the physics literature, as we will comment in Section~\ref{ssec: ssigma}, we present here a formula originally derived in Ref.\ \onlinecite{MaPaTe} (see Theorem~\ref{thm:notKubo}) which allows to compare different sensible proposals for the spin conductivity, especially in presence of Rashba-type spin-orbit interactions which break the conservation of the $z$-component of spin at equilibrium. The formula appears very different and more involved than the standard DCF for the (charge) Hall conductivity, exactly because of spin non-conservation. Rashba spin-orbit coupling is in particular responsible for the presence of a \emph{spin torque}, which prevents the spin current to satisfy a (sourceless) \emph{microscopic} continuity equation \cite{ShiZhangXiaoNiu06}. As a further application of the construction of the NEASS to all orders developed in Ref.\ \onlinecite{MarcelliMonaco}, and as an original contribution of this paper, we show that there is no \emph{mesoscopic} spin torque in the NEASS up to any order in $\eps$ (that is, the trace per unit volume which measures the expectation of the spin torque operator vanishes in the NEASS faster than any power of $\eps$).

Finally, we comment in Section~\ref{sec: ssigma=sG} on the relation between the intensive notion of \emph{conductivity} versus the extensive notion of \emph{conductance}, again for spin transport. While it is not difficult to see that their charge analogues agree in two dimensions (as one could argue directly from dimensional reasons) \cite{AvronSeilerSimon}, the equality of spin conductance and spin conductivity is far more non-trivial to prove, also because of the more involved expression for the spin conductivity mentioned above. A full investigation in this sense is still an interesting open line of research, but we report here on a preliminary analysis conducted in Ref.\ \onlinecite{MarcelliPanatiTauber19} on the ``Kubo-like'' terms in the formula for the spin conductivity, i.e.\ those which more resemble the DCF for the Hall conductivity. We show that those terms do in fact agree with appropriately-defined expressions which contribute to the spin conductance, hinting at a further analogy between spin and charge transport which we plan to study in the future.

\subsection{The model at equilibrium} \label{ssect: model}

After the above informal presentation of the aim of this paper, we start to describe some more mathematical features concerning the physical systems we address. While we choose to adopt a class of one-particle, periodic model Hamiltonians throughout the paper, many of the mathematical objects we will describe have been studied and generalized in other frameworks: in particular, for example, quantum transport and the Hall effect have been investigated in periodic crystals with and without a gap \cite{Cances_et_al} or in presence of (ergodic) disorder \cite{Bellissard, BoucletGerminetKleinSchenker05, ElgartSchlein04}, while the construction of the NEASS mentioned previously has been conducted also for models of interacting fermions on a lattice \cite{MonacoTeufel17, Teufel20}. At any rate, the reader is referred to Refs.\ \onlinecite{MaPaTe, MarcelliMonaco} and references therein for a more complete comparison with the existing literature and for many details which we will skip over in order to avoid technicalities and leave the presentation fairly non-technical.

\noindent
We consider a $2$-dimensional system which has a \emph{crystalline structure}, meaning that its configuration space $\X$ is invariant under translations by vectors in a Bravais lattice $\Gamma$. We will take into account both \emph{continuum models}, in which $\X=\R^2$ equipped with the Lebesgue measure, and \emph{discrete models}, in which $\X\subset\R^2$ is a discrete set of points (e.g.\ the square lattice $\Z^2$ or the honeycomb structure). One can assume, without loss of generality, that the Bravais lattice $\Gamma$ is spanned over the integers by a basis $\{\mathbf{a}_1, \mathbf{a}_2\}$ of $\R^2$.

\noindent
The one-particle Hilbert space for a particle moving on $\X$ and having spin-$1/2$ degrees of freedom is
\[
\Hi:=L^2(\X,\C^2).
\]
Notice how $\Hi$ comes equipped with natural position operators $X_1$ and $X_2$. The crystalline structure of configuration space is lifted to operators on $\Hi$ by means of \emph{(magnetic) translation operators} $T_\gamma$, $\gamma \in \Gamma$ (see Ref.\ \onlinecite{Zak64}).
In particular, an operator $A$ acting in $\Hi$ is called \emph{periodic} if and only if $[A,T_\gamma]=0$ for all $\gamma\in\Z^2$.

\noindent 
Crystalline systems, as we will see shortly, are described well by periodic operators. Consequently, physically relevant quantities such as expectation of observables in appropriate states are extensive, that is, they scale proportionally to the size of the region in which they are localized. Therefore, the appropriate notion suited for example to compute expectations of translation-invariant or more generally covariant or ergodic operators is the \emph{trace per unit volume}, which we now briefly introduce. 

The trace per unit volume of an operator $A$ is defined as 
\begin{equation}
\label{eqn:defn tau}
\tau(A):=\lim_{\substack{L\to\infty\\L\in 2\N+1}}\frac{1}{\abs{\FC_L}}\Tr(\chi_L A \chi_L).
\end{equation}
In the above, we denoted by $\chi_L$ the multiplication operator by the characteristic function of the cell $\FC_L$ centered around the origin, defined for $L\in 2\N +1$ as
\begin{equation}
\label{eqn:defn FC1}
\FC_L:=\left\lbrace x\in \X : x = \alpha_1 \, \mathbf{a}_1 + \alpha_2 \, \mathbf{a}_2 \text{ with }  |\alpha_j|\leq L/2 \; \forall\: j \in \set{1,2}\right\rbrace.
\end{equation}
It is clear that, for the above definition of $\tau(A)$ to make sense, the operator $A$ has to be such that $ \chi_L A \chi_L$ is trace class, and the limit defining $\tau(A)$ has to exist. In particular, operators $A$ such that $\tau(\abs{A}) < \infty$ are said to be of \emph{trace-per-unit-volume class}. Relevant classes of operators for our analysis will be periodic operators, and operators of the form $X_j \, A$ with $A$ periodic (compare e.g.\ \eqref{eqn:sigmaR}). The relevant properties of the trace per unit volume of such operators are listed below; a proof of these results can be found in Ref.\ \onlinecite[Propositions 2.4 \& 2.5]{MaPaTe} and Ref.\ \onlinecite[Lemma 3.22]{BoucletGerminetKleinSchenker05}.

\begin{proposition} 
\label{prop:tau}
Let $A$ be periodic and such that $ \chi_L A \chi_L$ is trace class. Then
\begin{enumerate}[label=(\roman*), ref=(\roman*)]
\item \label{it:tauper} $\tau(A)$ is well-defined and 
\[
\tau(A) =\frac{1}{\abs{\FC_1}} \Tr(\chi_1 A \chi_1).
\]
\item \label{it:tauAX} For $j \in \set{1,2}$ the operator $X_j\,A$ has finite trace per unit volume and 
\[
\tau(X_j \,A ) = \frac{1}{\abs{\FC_1}}\Tr\left( \chi_1 X_j\, A \chi_1 \right).
\]
\item \label{it:tauAXex} If in addition  $\tau(A)=0$, then $\tau(X_j \,A)$ does not depend on the choice of the origin, in the sense that 
\[
\tau((X_j +\alpha) A )=\tau(X_j  \,A )\quad\forall\,\alpha\in\R.
\]
In particular, $\tau(X_j \,A)$ does not depend on the choice of the center of the cell $\FC_L$ in \eqref{eqn:defn FC1}.
\item \label{lem: ctau} \emph{(Cyclicity of the trace per unit volume)} Given $A,B$ periodic operators such that $A$ is bounded and $B$ is such that $\tau(\abs{B})<\infty$ then $\tau(AB)=\tau(BA).$
\end{enumerate}
\end{proposition}

After this detour on the trace per unit volume, we are finally able to formulate the type of quantum system on the crystalline configuration space $\X$ that we want to describe. Equilibrium properties of the system will be encoded in an Hamiltonian $H_0$ which is hereonafter required to satisfy the following

\begin{assumption} 
\label{assum:H0}  
\begin{enumerate}[label=$(\mathrm{H}_{\arabic*})$,ref=$(\mathrm{H}_{\arabic*})$]
\item \label{item:H0}
The Hamiltonian $H_0$ of the unperturbed system is a periodic self-adjoint operator acting in $\Hi$ and bounded from below;

\item \label{item:smooth H0}
the fibers $H_0(k)$ in the Bloch--Floquet--Zak (BFZ) representation depend smoothly on $k$ in an appropriate topology (compare Ref.\ \onlinecite[Sect. 3]{MaPaTe});

\item \label{item:gap} the Fermi energy $\mu\in\R$ is in a spectral gap of $H_0$. We denote by $\Pi_0 = \chi_{(-\infty, \mu)}(H_0)$ the corresponding spectral projection (Fermi projection), which is supposed to be \emph{trace-per-unit-volume class}. \hfill $\diamond$
\end{enumerate}
\end{assumption}

\noindent
The above assumptions are fulfilled by a wide class of physically relevant models, both on continuum and on discrete configuration spaces. The reader may in particular refer to Bloch--Landau operators (under mild regularity assumptions on the electro-magnetic potentials --- see e.g.\ Ref.\ \onlinecite[Sect. 3]{MaPaTe}) employed as continuum models, or to tight-binding Hamiltonians of common use in condensed matter physics such as the Haldane \cite{Haldane88, MarcelliMonacoMoscolariPanati} and Kane--Mele \cite{KaneMele2005, MarcelliPanatiTauber19} models on the honeycomb lattice.

\noindent 
For concreteness, we will at times focus on model Hamiltonians which satisfy the simpler but stronger

\begin{assumption} \label{stronger}
The Hamiltonian $H_0$ is a spectrally-gapped self-adjoint periodic operator, bounded from below, acting on a discrete configuration space $\X \simeq \FC_1 \times \Gamma$ and having \emph{finite range hopping amplitudes}, that is, for some $R \ge 0$
\[ (H_0)_{\mathbf{m},\mathbf{n}} = 0 \quad \text{if } |\mathbf{m}-\mathbf{n}| > R \]
where $A_{\mathbf{m},\mathbf{n}}$ refers to the ``matrix element'' (hopping amplitude) of the operator $A$ with respect to the position coordinates $\mathbf{m},\mathbf{n} \in \Gamma$. As before, we denote by $\Pi_0$ the Fermi projection onto the (finitely many) energy bands below the spectral gap.
\end{assumption}

Most tight-binding Hamiltonians, including the ones mentioned above, are indeed covered by this Assumption: nearest-neighbour Hamiltonians correspond to $R=1$, while next-to-nearest-neighbour Hamiltonians correspond to $R=2$.

%%%%%% SECTION 1 %%%%%%%%%%%%%%%%%%%%%%%%%%%%%

\section{Linear response formalism}

\subsection{A heuristic introduction} \label{sec:heuristic}
This section provides a heuristic and informal illustration of the linear response formalism. Our presentation follows the monograph Ref.\ \onlinecite{AizenmanWarzel15} and a recent topical review Ref.\ \onlinecite{HenheikTeufel21}, to which we refer the reader for an in-depth analysis of the subject.

\noindent
In the framework of Hamiltonian quantum systems, the linear response formalism addresses the following problem. For concreteness, assume that the system is governed by a Hamiltonian operator $H_0$ of the type mentioned above and prepared in a equilibrium state specified by the Fermi projection $\Pi_0$. The choice of this equilibrium state is motivated by the fact that the physical applications under consideration take place at temperatures close to zero. Then a small perturbation $\eps A$ is adiabatically switched on, where $0\leq\eps\ll 1 $. Thus, the time evolution is described by a time-dependent self-adjoint operator on some Hilbert space of the form
\begin{equation}
\label{eqn: Heps(t)}
H^\eps(t):=H_0+\eps f(\eta t) A,\qquad \text{for } t\in \R,\, 0<\eta\ll 1,
\end{equation}
where $f$ is a \emph{switching function}, i.e.\ $f\colon \R\to [0,1]$ is a smooth map such that $f(s)=0$ for $s\leq -1$ and $f(s)=1$ for $s\geq 0$. Notice that the perturbation is turned on in the finite interval $[-1/\eta,0]$ and the process becomes adiabatic when $\eta\to 0^+$. 

One is interested in determining the state of the perturbed system $\rho^\eps$, singling out the linear term in $\eps$ in its expression, namely $L$ in the following equality:
\begin{equation}
\label{eqn: rho_eps}
\rho^\eps=\Pi_{0}+\eps \,L + o(\eps)\qquad \text{as $\eps\to 0$.}
\end{equation}
A natural candidate for $\rho^\eps$ is the solution of the Cauchy problem associated with $H^\eps(t)$ at $t=0$ (or any non-negative time, when the perturbation is completely switched on), i.e.\ $\rho^\eps:=\rho^\eps(0)$ where $\rho^\eps(t)$ solves\footnote{We set $\hbar=1$ henceforth. Also, since $f(\eta t)$ vanishes for every $t\leq -1/\eta$, actually one can set the initial datum $\rho^\eps(t_0) = \Pi_0$ at any $t_0\leq -1/\eta$.}
\begin{equation} 
\label{eqn: Cauchy pb}
\begin{cases}
\iu \frac{\di}{\di t}{\rho}^\eps (t) = [H^\eps (t), \rho^\eps(t)] \\
\rho^\eps(-\infty) = \Pi_0.
\end{cases}
\end{equation}
The usual paradigm for finding an expression for $\rho^\eps$ as in \eqref{eqn: rho_eps} is based on standard time-dependent perturbation theory. This commonly used approach leads to the Kubo formula \cite{Kubo57} for the \emph{linear response coefficient} $\Delta_{B,A}$ of a generic observable operator $B$ caused by a perturbation $A$. Here we recall briefly its derivation. First one chooses a particular time-dependent function for modeling a  switching process\footnote{\label{ft: fkubo}Note that the following map $f\sub{Kubo}$ does not satisfy the properties required to be a switching function, as it is not smooth but continuous and vanishes only as $t\to -\infty$.}, namely $f\sub{Kubo}(\eta t)=\eu^{\eta t}\chi_{(-\infty,0]}(t)+\chi_{(0,\infty)}(t)$, and then, going to \emph{interaction picture}, one singles out the leading order in $\eps$ in the generator of the perturbed dynamics. Indeed, formulating the Cauchy problem in $\rho^\eps\sub{int}(t):=\mathrm{e}^{\iu t H_0}\rho^\eps(t)\mathrm{e}^{-\iu t H_0}$ equivalent to \eqref{eqn: Cauchy pb}, one obtains that
\begin{equation} 
\label{eqn: int Cauchy pb}
\begin{cases}
\iu\frac{\di}{\di t}\rho^\eps\sub{int}(t)= \eps f\sub{Kubo}(\eta t) [\mathrm{e}^{\iu t H_0}A\mathrm{e}^{-\iu t H_0},\rho^\eps\sub{int}(t)]\\
\rho^\eps\sub{int}(-\infty)=\Pi_0.
\end{cases}
\end{equation}
By applying twice the fundamental theorem of calculus, one has that
\begin{align*}
\rho^\eps(0)-\Pi_0&=\rho^\eps\sub{int}(0)-\Pi_0=-\iu\,\eps \int_{-\infty}^{0}\di s\, f\sub{Kubo}(\eta s) [\mathrm{e}^{\iu s H_0}A\mathrm{e}^{-\iu s H_0},\rho^\eps\sub{int}(s)]\\
&=\iu\eps \int_{-\infty}^{0}\di s\, \eu^{\eta s} \mathrm{e}^{\iu s H_0}[\Pi_0,A]\mathrm{e}^{-\iu s H_0}+\eps^2 R,
\end{align*}
therefore \emph{formally} (as we have not specified in which sense $o(\eps)$ in \eqref{eqn: rho_eps} must be understood)
\begin{equation}
\label{eqn: LKubo}
L\su{Kubo}=\iu\int_{-\infty}^{0}\di s\, \eu^{\eta s} \mathrm{e}^{\iu s H_0}[\Pi_0,A]\mathrm{e}^{-\iu s H_0},
\end{equation}
which is called the \emph{linear response ansatz} (LRA) for the perturbed state --- see Ref.\ \onlinecite[\S 13.2]{AizenmanWarzel15}. Thus, already at this stage the following mathematical problem emerges: justify the LRA, that is, show if \eqref{eqn: LKubo} is a first-order approximation in $\eps$ to the actual solution of Cauchy problem \eqref{eqn: Cauchy pb}. More specifically, 
since $\rho^{\eps}\equiv\rho^{\eps,\eta}$, one has to prove that 
\begin{equation}
\label{eqn: just LRA}
\lim_{\eps\to 0^+}\lim_{\eta\to 0^+}\frac{\vertiii{\rho^{\eps,\eta}-\Pi_0-\eps L\su{Kubo}}}{\eps}=0,
\end{equation}
where $\vertiii{\,\cdot\,}$ is the distance in an appropriate topology. Showing the above equality in the operator norm proves to be a formidable task, and therefore one is prompted to adopt a weaker topology. A common choice requires to prove a similar statement testing it against observable operators, which is also physically relevant since these expectation values can be interpreted as physical measurements. To be more precise, for every observable $B$ one is interested in proving the following expansion:
\begin{equation}
\label{eqn: deltaBA}
\mathcal{T}(B\rho^\eps)-\mathcal{T}(B\Pi_0)=:\Delta_{B,A}\,\eps +o(\eps)\qquad \text{as $\eps\to 0$},
\end{equation}
where $\mathcal{T}(\,\cdot\,)$ is a suitable trace-like functional (e.g.\ for extended systems it could be the trace per unit volume, as discussed in the Introduction). 
The term $\Delta_{B,A}$ is called the \emph{linear response coefficient of $B$ induced by $A$}.
By using expression \eqref{eqn: LKubo}, the LRA becomes the statement that $\Delta_{B,A}$ in \eqref{eqn: deltaBA} equals
\begin{equation}  \label{eqn: DeltaAB Kubo}
\Delta_{B, A}\su{Kubo} = \lim_{\eta \to 0^+} \Delta_{B, A}\su{Kubo}(\eta) \quad \text{with} \quad
\Delta_{B, A}\su{Kubo}(\eta):=\mathcal{T}(B L\su{Kubo})=\iu\int_{-\infty}^{0}\di s\, \eu^{\eta s} \mathcal{T}\left(  B\,\mathrm{e}^{\iu s H_0}[\Pi_0,A]\mathrm{e}^{-\iu s H_0}\right).
\end{equation}

\subsection{Construction of a non-equilibrium almost-stationary state}

Even for this ``weak'' version of the LRA, the main difficulties in its mathematical treatment are rooted in the time-adiabatic limit: 
indeed what one seeks is a tractable characterization of the state $\rho^{\eps,\eta}$ in the limit $\eta\to 0^+$.
Thus, one might wonder:

\medskip

\noindent
\textbf{Question (Q$\mathbf{_1}$)}
Is it possible to circumnavigate the standard time-adiabatic perturbation method in order to identify a non-equilibrium state at $0<\eps\ll 1$ leading to the computation of the linear response coefficient $\Delta_{B,A}$ as in \eqref{eqn: deltaBA}?

\medskip

A possible answer to the above question in the context of quantum transport in gapped systems has been recently proposed by S. Teufel and collaborators \cite{Teufel20, HenheikTeufelI, HenheikTeufelII}, also jointly with the authors \cite{MonacoTeufel17, MaPaTe, MarcelliMonaco, MarcelliTeufel}, in the form of the so-called \emph{non-equilibrium almost-stationary state} (NEASS). The NEASS paradigm relies on ideas conceived in previous works on the space-adiabatic perturbation theory in quantum dynamics \cite{PanatiSpohnTeufel02,PanatiSpohnTeufel03, Teufel03}. 

\noindent
From here on out, we adopt the mathematical setting specified in Section~\ref{ssect: model} for the equilibrium Hamiltonian $H_0$, given by a gapped periodic one-particle operator. The same framework has been adopted in Ref.\ \onlinecite[Sect. 4]{MaPaTe}, where the construction of the NEASS which we will illustrate shortly has been obtained up to the first order in $\eps$, as well as in Ref.\ \onlinecite[Sect. 3]{MarcelliMonaco}, where the same construction was generalized to every order in $\eps$.

Our goal is to investigate the linear response of charge or spin current to an external electric field of small intensity (compare Section \ref{sec: sigma} below). Thus, given the equilibrium Hamiltonain $H_0$ as in Assumption~\ref{assum:H0} or \ref{stronger}, we specify a \emph{stationary} perturbation which drives the system out of equilibrium by adding an external electric potential in the $2$-nd direction: 
\begin{equation}
\label{eqn: defn Heps}
H^\eps:=H_0-\eps X_2,
\end{equation}
where $\eps\in [0,1]$ and $X_2$ is the position operator in the $2$-nd direction.

\noindent 
For every $n\in\N$, the NEASS $\Pi^\eps_n$ which well approximates the non-equilibrium dynamics up to order $\eps^{n+1}$ is  
characterized by the following Theorem, whose proof can be found in Ref.\ \onlinecite[Sect.\ 3]{MarcelliMonaco}.

\begin{theorem}
\label{thm:constr NEASS}
Let $H_0$ and $\Pi_0$ satisfy Assumption~\ref{assum:H0} and consider the Hamiltonian \eqref{eqn: defn Heps}. Then, there exists a sequence of bounded and periodic operators $\{A_j\}_{j\in\N}$ (whose BFZ fibers are smooth functions of $k$ in an appropriate topology) such that, setting for any $n\in \N$
\begin{equation}
\label{eqn:defn Sneps}
\SC_n^{\eps}:=\sum_{j=1}^n \eps^{j-1} A_j\, \text{ for $n\geq 1$}\quad \text{ and }\quad \SC_0^{\eps}:=0\, ,
\end{equation}
we have that
\begin{equation}
\label{eqn: NEASS prop}
\Pi^\eps_n := \eu^{\iu \eps \SC_n^{\eps}} \,\Pi_0\, \eu^{-\iu \eps \SC_n^{\eps}}\quad \text{ satisfies }\quad
[H^{\eps} , \Pi^{\eps}_n]=\eps^{n+1}[R^\eps_n,\Pi^{\eps}_n]
\end{equation}
where $R^\eps_n$ is bounded, periodic and its operator norm $\norm{R^\eps_n}$ is uniformly bounded in $\eps\in [0,1]$.
\end{theorem}

It should be emphasized that the $A_j$'s appearing in the statement above can be explicitly computed, and lead in particular to the identification of the coefficients in the following expansion of the NEASS in powers of $\eps$:
\begin{equation}
\label{eqn: NEASS exp}
\Pi^\eps_n = \Pi_0 + \eps\,\Pi_1+\eps^2 \,\Pi_2 + \cdots + \eps^n \, \Pi_n + \eps^{n+1} \Pi\sub{reminder}^\eps. 
\end{equation}
For example, the first-order operator $\Pi_1$ appearing in the above expression (which plays a crucial role in the computation of linear response coefficients, as we will see below) is explicitly given by
\begin{equation}
\label{eqn: defn Pi1}
\Pi_1=\LI_{H_0}^{-1}\big({[  X_2,\Pi_0]} \big),
\end{equation}
where $\LI_{H_0}^{-1}$ is the inverse of the Liouvillian (super-)operator $A \mapsto \LI_{H_0}(A):=[H_0,A]$. Moreover, whenever well defined (that is, away from the kernel of $\LI_{H_0}$), the inverse of the Liouvillian has the integral representation
\begin{equation*}
%\label{eqn:I(A)}
\LI_{H_0}^{-1}(A) := \frac{\iu}{2 \pi} \oint_C \di z \, (H_0 - z \Id)^{-1} \, [A,\Pi_0] \, (H_0 - z \Id)^{-1},
\end{equation*}
where $C$ is a positively-oriented contour in the complex energy plane enclosing the part of the spectrum of $H_0$ below the Fermi energy $\mu$ (compare \ref{item:gap} in Assumption \ref{assum:H0}). Notice that the gap condition ensures that $\LI_{H_0}^{-1}(A)$ is a bounded operator if $A$ is.

\noindent
We stress that the above construction of the NEASS does not require any time-dependent perturbation of the Hamiltonian $H_0$. However, it still provides a good approximation of the adiabatically time-evolved state, described in the previous section. 
Indeed, one can prove \cite{MarcelliTeufel} that for all times $t \geq 0$ and any non-zero $n,m\in \N$ 
\begin{equation} 
\label{eqn: apprNEASS}
\sup_{\eta\in I_{m,\eps}} \left| \tau( A \, \rho^{\eps}(t)) - \tau( A \, \Pi_n^\eps) \right| = \Or(\eps^{n+1})   
\end{equation}
uniformly on bounded intervals in (macroscopic) time, where $\rho^{\eps}(t)\equiv \rho^{\eps,\eta,f}(t)$ is the solution of the Cauchy problem \eqref{eqn: Cauchy pb} and $A$ is a \emph{suitable} observable (e.g.\ charge or spin current operator). Here
\[
I_{m,\eps}= [\eps^m ,\eps^{1/m}]\text{ is called an \emph{interval of admissible time-scales for $\eta$}}.
\]
Too slow switching ($\eta\ll \eps^m$ for all $m\in \N^{*}$) 
must be excluded, because due to tunneling the NEASS decays on such long times-scales, while too fast switching ($1\gg \eta\gg \eps^{1/m}$ for all $m\in \N^*$) would only give an error $o(1)$ on the right-hand side of~\eqref{eqn: apprNEASS} (see Refs.\ \onlinecite{Teufel20, HenheikTeufel21}).

%%%%% SECTION 2 %%%%%%%%%%%%%%%%%%%%%%%%%%%%%%
\section{Charge and spin conductivity}
\label{sec: sigma}

Now that the NEASS construction has been established in reference to Question \textbf{(Q$\mathbf{_1}$)}, we come back to the investigation of linear response of observables to a perturbation of the form \eqref{eqn: defn Heps}.

\noindent
Specifically, we will be interested in probing a charge (resp.\ spin) current in the $1$-st direction. In this case the linear response coefficient is the \emph{transverse charge} (resp.\ \emph{spin}) \emph{conductivity} 
\begin{equation}
\label{eqn: sigma12sharp}
\sigma_{12}^\sharp:=\Delta_{J_1^\sharp,\,- X_2 },\qquad\text{ with } \sharp \in\{ \mathrm{c}, \mathrm{s}\},
\end{equation}
where $\mathrm{c}$ (resp.\ $\mathrm{s}$) stands for charge (resp.\ spin), and the right-hand side has to be understood in the sense of equality \eqref{eqn: deltaBA} with $B=J_1^\sharp$ being the charge (resp.\ spin) current operator in the $1$-st direction, $A=-X_2$ being the linear potential inducing the electric field in the $2$-nd direction, and the trace-like functional $\mathcal{T}$ is the trace per unit volume $\tau$ defined in \eqref{eqn:defn tau}.

\subsection{Charge conductivity}
\label{ssec: csigma}
In the case of $\sigma_{12}\su{c}$, which is relevant for the quantum Hall effect as explained in the Introduction, the current operator takes the form $J_1\su{c}:=\iu [H_0, X_1]$ (assuming the charge carriers have unit charge), whose definition is well-known to be justified by a continuity equation for the charge transport.

\noindent
In this context, our main result is given by the following theorem, proven in Ref.\ \onlinecite[Sect. 4]{MarcelliMonaco} (see also Refs.\ \onlinecite{KleinSeiler,Bachmannetal} for previous related results).

\begin{theorem}
\label{thm: sigmahall}
Under the assumptions and in the notation of Theorem \ref{thm:constr NEASS}, for every $ n\in\N$ we have that
\[
\tau(J_1\su{c} \,\Pi^\eps_n)=\eps\,\sigma\sub{Hall}+\Or(\eps^{n+1}),
\]
where
\begin{equation}
\label{eqn: sigmahall}
\sigma\sub{Hall}:=\iu \tau(\Pi_0\left[[\Pi_0,X_1],[\Pi_0,X_2]\right]).
\end{equation}
\end{theorem}

\noindent 
The expression \eqref{eqn: sigmahall} is also dubbed a double-commutator formula (DCF). Notice that it depends only on the equilibrium state $\Pi_0$ and the position operators in both directions: $X_1$, related to the observable $J_1\su{c}$, and $X_2$ associated with the perturbation. The computation of $\sigma\sub{Hall}$ is determined by $\Pi_1$, being the term carrying the linear factor in the expansion \eqref{eqn: NEASS exp}. 

In this context, the construction of the NEASS has a twofold advantage: First, it's sufficiently explicit to allow to compute the charge conductivity $\sigma_{12}\su{c}$; second, it provides a justification of the LRA through the relation \eqref{eqn: apprNEASS}, without the need to compute (and make sense of) the limit \eqref{eqn: DeltaAB Kubo}.

\noindent
Finally, notice that the above Theorem states that the conductivity $\sigma_{12}\su{c}$ is given by the Hall conductivity $\sigma\sub{Hall}$ up to \emph{arbitrarily high orders} in the strength of the perturbing electric field. Therefore, this result also yields the validity of the Kubo formula beyond linear response.

\subsection{Spin conductivity}
\label{ssec: ssigma}

We now shift our attention to spin transport, aiming at the derivation of general formulas for $\sigma\su{s}_{12}$ as in \eqref{eqn: sigma12sharp}.

\noindent
In contrast with the study of charge transport which was detailed the previous subsection, the mathematical theory of spin transport is still at a preliminary stage. First, despite almost two decades of scientific debate, no common agreement has been achieved yet about the correct form of the operator representing the spin current density. The discussion arises in situations in which the unperturbed Hamiltonian operator $H_0$ does not commute with the spin operator: this can occur in physically relevant systems in presence of so-called Rashba spin-orbit interactions, as in the emblematic model for the quantum spin Hall effect introduced by Kane and Mele \cite{KaneMele2005} (see also Ref.\ \onlinecite[Appendix A]{MarcelliPanatiTauber19} for its first-quantized form). The lack of spin conservation prevents the validity of a (sourceless) microscopic continuity equation for spin current \cite{ShiZhangXiaoNiu06}, thus obscuring the very notion of a quantum-mechanical spin current operator.

To be more precise, fix a direction of observation for the spin, say $z$ to be specific. Let $S_z$ be the operator representing the $z$-component of the spin, defined in the framework described in Section~\ref{ssect: model} as
\[
S_z:=\Id_{L^2(\X)}\otimes s_z,\quad\text{where $s_z = \sigma_z/2$ is half of the third Pauli matrix}.
\]
\goodbreak
One can consider the following two definitions for a spin current operator in the $1$-st direction:
\begin{itemize} 
\item[$-$] 
the \emph{conventional} spin current operator
\begin{equation} 
\label{eqn: J_conv}
 J\su{s}\sub{conv,1}:= 
\frac{1}{2} \big( \iu [H_0, X_1] \, S_z + \iu  S_z \, [H_0, X_1] \big)
\end{equation}
which has been adopted e.g.\ in Refs.\ \onlinecite{Sinovaetalii, ShengShengTingHaldane2005, Schulz-Baldes2013};
\item[$-$] the \emph{proper} spin current operator
\begin{equation} 
\label{eqn: J_proper}
 {J}\sub{prop,1}\su{s} := \iu [H_0, X_1 S_z]
\end{equation}
proposed in Refs.\ \onlinecite{ShiZhangXiaoNiu06, cinesi} and used in Ref.\ \onlinecite{MarcelliPanatiTauber19}.
\end{itemize} 
Clearly, whenever $[H_0,S_z]=0$ there is no difference between the two definitions above, and spin transport ``boils down to two copies'' of charge transport, filtering according to the spin sector. On the other hand, in general $[H_0, S_z]\neq 0$ due to Rashba interactions, thus an ambiguity emerges in the choice between ${J}\sub{conv,1}\su{s}$ and ${J}\sub{prop,1}\su{s}$. 

Let us analyze the main differences between these two operators in the spin non-conserved case, under Assumption~\ref{assum:H0}:
\begin{itemize}
\item[$-$] $J\su{s}\sub{conv,1}$ is periodic, is not expressed as a full commutator and in discrete models it is also a bounded operator, in view of the underlying ultraviolet cutoff;
\item[$-$] ${J}\sub{prop,1}\su{s}$ is not periodic, is in the form of a full commutator and is not a bounded operator.
\end{itemize}
The two properties (periodicity and being in the form of a full commutator) for the current operator enter crucially in the derivation via the Kubo formula of the associated conductivity, as mentioned in Section \ref{sec:heuristic} above (see Ref.\ \onlinecite[Appendix A]{AizenmanGraf98} and Ref.\ \onlinecite[Chap. 1]{Marcelli18} for more details). Since neither of the two spin current operators proposed above possesses both of these properties (if spin is not conserved), the derivation of a Kubo formula for spin conductivity becomes cumbersome (compare also the discussion in Section \ref{sec: ssigma=sG} below). 

Therefore, the following questions naturally arise:

\medskip

\noindent
\textbf{Question (Q$\mathbf{_{2.a}}$)}
In the general case $[H_0,S_z]\neq 0$, is it possible to derive formulas for the transverse spin conductivity $\sigma_{12}\su{s}$? 

\medskip

\noindent
\textbf{Question (Q$\mathbf{_{2.b}}$)}
If \textbf{(Q$\mathbf{_{2.a}}$)} has a positive answer, how will these formulas be affected by the choice of ${J}\sub{conv,1}\su{s}$, respectively ${J}\sub{prop,1}\su{s}$, as a spin current operator?

\medskip 

\noindent
To stress the dependence of the spin conductivity appearing in \eqref{eqn: sigma12sharp} on the choice of spin current operator, we denote by%
\footnote{The specification of the real part is needed due to the lack of periodicity for the proper spin current operator $J\sub{prop,1}\su{s}$. In fact, in general the cyclicity of $\tau$ does not hold for non-periodic operators (compare Proposition~\ref{prop:tau}\ref{lem: ctau}) and consequently one can not deduce that the trace per unit volume defining the conductivity is real-valued, even if the operators involved are self-adjoint.}
\begin{equation}
\label{eqn: defn sigmac/p12}
\sigma_{\sharp,12}\su{s}:=\re\Delta_{{J}_{\sharp,1}\su{s},\,- X_2}\quad \text{for } \sharp \in \set{\text{conv, prop}}.
\end{equation}

\noindent
To address the questions raised above, and in view of the above-mentioned difficulties of formulating the LRA and Kubo formula in the context of spin transport, we exploit once again the NEASS construction. 
The next Theorem, proven in Ref.\ \onlinecite{MaPaTe}, answers both Questions \textbf{(Q$\mathbf{_{2.a}}$)} and \textbf{(Q$\mathbf{_{2.b}}$)}. 

\begin{theorem}
\label{thm:notKubo}
Under the assumptions and in the notation of Theorem \ref{thm:constr NEASS}, let 
\[ \Pi_1^\eps  = \Pi_0 + \eps \, \Pi_1 + \Or(\eps^2) \]
be the NEASS at first order (compare \eqref{eqn: NEASS exp}). 
Then
\begin{equation}
\label{eqn:sigmaprco}
\sigma\subm{prop}{12}\su{s}= \sigma\subm{conv}{12}\su{s}+\sigma\subm{rot}{12}\su{s},
\end{equation}
where 
\begin{equation}
\label{eqn:sigmaT}
\begin{aligned}
\sigma\subm{conv}{12}\su{s}&=\re \tau \Big(  \iu \Pi_0\big[[X_1,\Pi_0]S_z, [X_2,\Pi_0] \big] \Big) \\
& \quad + \re \tau  \Big( \iu [H_0,X_1\su{D}] S_z\su{OD} \Pi_1 + \iu X_1\su{OD} [S_z,H_0]\Pi_1   \Big) 
\end{aligned}
\end{equation}
and the rotation contribution to the proper spin conductivity is defined as
\begin{equation}
\label{eqn:sigmaR}
\sigma\subm{rot}{12}\su{s}:=\re \tau(\iu X_1  [H_0, S_z ]  \Pi_1).
\end{equation}
In \eqref{eqn:sigmaT}, we have denoted for every operator $A$
\begin{align*}
A\su{D}& := \Pi_0 A \Pi_0 + \Pi_0^\perp A \Pi_0^\perp\,\, \text{ its diagonal part, and}\\
A\su{OD}&  := \Pi_0 A \Pi_0^\perp + \Pi_0^\perp A \Pi_0\,\, \text{ its off-diagonal part}
\end{align*}
with respect to the orthogonal decomposition of the Hilbert space $\Hi$ into $\Ran \Pi_0 \oplus (\Ran \Pi_0)^\perp$.
\end{theorem}

\noindent
While it appears that in general the conventional and proper spin conductivities differ in view of the rotation contribution \eqref{eqn:sigmaR}, it could still be that in certain specific models the latter vanish. This question turns out to be related to  
the so-called {\it Unit Cell Consistency} (UCC), 
namely the natural requirement that any prediction on macroscopic transport must be independent of the choice of the fundamental (unit) cell. 
The dependence of the expressions in \eqref{eqn:sigmaT} and \eqref{eqn:sigmaR} on the particular fundamental cell $\FC_1$ chosen is implicitly given by the trace per unit volume (see \eqref{eqn:defn tau}). 
While $\sigma\subm{conv}{12}\su{s}$ satisfies UCC due to the periodicity of $J\sub{conv,1}\su{s}$, in general $\sigma\subm{prop}{12}\su{s}$ does not because of the non-periodic operator $\iu X_1 [H_0, S_z ]  \Pi_1$ appearing in \eqref{eqn:sigmaR}. At this point it is useful to recall the following well-known fact from solid state physics: given any two fundamental cells of arbitrary shape for the Bravais lattice $\Gamma$, it is possible to cut the first one up into a finite number of pieces $P_\gamma$ with $\gamma\in I\subset\Gamma$ such that, when $P_\gamma$ are translated through suitable lattice vectors, they can be reassembled to give the second. A sufficient condition to guarantee that $\sigma\subm{prop}{12}\su{s}$ fulfills UCC is that 
\begin{equation}
\label{eqn: suff cond UCC}
\Tr(\chi_{P_\gamma}\,\iu  [H_0, S_z ]  \Pi_1\,\chi_{P_\gamma})=0\quad \text{ for all $\gamma\in I$}.
\end{equation}

Actually, if one restricts to a class of discrete models for which the previous vanishing condition holds true, then one obtains the equality of conventional and proper spin conductivity.

\begin{proposition}
\label{prop:prsigma=cosigma}
Let $H_0$ and $\Pi_0$ be as in Assumption \ref{stronger}. Assume that the condition \eqref{eqn: suff cond UCC} is satisfied and that $\mathrm{Rank}{\chi_{P_\gamma}}=1$. Then 
$$\sigma\subm{rot}{12}\su{s}=0\quad\text{or equivalently}\quad\sigma\subm{prop}{12}\su{s}=\sigma\subm{conv}{12}\su{s}.$$
\end{proposition}

\noindent
Let us point out that any model, enjoying a discrete rotational symmetry, satisfies  condition \eqref{eqn: suff cond UCC} (see Ref.\ \onlinecite[Proposition A.3]{MaPaTe}). Remarkably, the Kane--Mele model is in this class. Thus, in this class of discrete models fulfilling a discrete rotational symmetry, despite at the operator level $J\sub{prop,1}\su{s}\neq J\sub{conv,1}\su{s}$, at the \emph{expectation value} level (that is, in the sense of the corresponding conductivities) there is no difference. One can construct models in which this further symmetry is broken and which then produce different values for $\sigma\sub{prop,1}\su{s}$ and $\sigma\sub{conv,1}\su{s}$ \cite{MonacoUlcakar20}.

As a side note, let us mention that, while here we have addressed pertubations of the Hamiltonian via electric potentials, the properties of the (spin) Hall conductivity could be investigated also as the response of a (spin) current to the variation of an external uniform \emph{magnetic} field. This approach leads to the so-called St\v{r}eda formula for the Hall conductivity \cite{Streda}. While the derivation of the St\v{r}eda formula has been understood mathematically for charge transport \cite{CorneanNenciuPedersen}, its generalization to spin transport is still at a preliminary stage: in this sense, we refer the reader to Ref.\ \onlinecite{MonacoMoscolari}, where the methods of Ref.\ \onlinecite{CorneanNenciuPedersen} have been adapted to discuss the spin conductivity in the case where $[H_0, S_z]=0$.

\subsection{Vanishing of spin torque}
As we discussed above, much of the debate in the description of spin currents originates from the presence in the Hamiltonian of terms which do not commute with spin. However, since the perturbation by an electric field does not directly couple to the spin degrees of freedom ($[X_2, S_z]=0$), one might expect from a physical point of view that a uniform electric field does not induce any particular spin torque excess or deficiency in the sample. To illustrate the techniques explained thus far, in this Subsection we show exactly this claim, by proving that the \emph{spin torque operator} $\iu [H_0,S_z]$ has an expectation in the NEASS which vanishes at any order in $\eps$, thus corroborating the claim in Ref.\ \onlinecite{ShiZhangXiaoNiu06} that ``\emph{one is often interested in a particular component of the spin, and the corresponding torque component can vanish in the bulk on average [\dots] This is certainly true for the many models used for the study of the spin-Hall effect}''.

\begin{theorem} \label{thm: spin torque}
Under the assumptions and in the notation of Theorem \ref{thm:constr NEASS}, we have for every $n \in \N$ that
\[ \tau \left( \iu [H_0, S_z] \, \Pi^{\eps}_n \right) = \Or(\eps^{n+1}). \]
\end{theorem}
\begin{proof}
We heavily rely on results proved in Ref.\ \onlinecite{MarcelliMonaco}. First of all, observe once again that $\iu [H_0, S_z] = \iu [H^\eps, S_z]$, as the linear potential $-\eps X_2$ commutes with spin. We therefore have
\[ \tau \left( \iu [H_0, S_z] \, \Pi^{\eps}_n \right) = \tau \left( \iu [H^\eps, S_z] \, \Pi^{\eps}_n \right) = \tau \left( \iu \left[\Pi^{\eps}_n \, H^\eps  \, \Pi^{\eps}_n, \Pi^{\eps}_n \, S_z  \, \Pi^{\eps}_n\right] \right) + \Or(\eps^{n+1}), \]
where we used cyclicity of the trace per unit volume (which is definitely satisfied in view of periodicity and of the trace-per-unit-volume class nature of the projection $\Pi^\eps_n$) and the fact that $[H^\eps, \Pi^\eps_n] = \Or(\eps^{n+1})$ in view of Theorem~\ref{thm:constr NEASS}. We now use the explicit form of $H^\eps = H_0 - \eps X_2$ to obtain
\[ \tau \left( \iu [H_0, S_z] \, \Pi^{\eps}_n \right) = \tau \left( \iu \left[\Pi^{\eps}_n \, H_0  \, \Pi^{\eps}_n, \Pi^{\eps}_n \, S_z  \, \Pi^{\eps}_n\right] \right) - \eps \, \tau \left( \iu \left[\Pi^{\eps}_n \, X_2 \, \Pi^{\eps}_n, \Pi^{\eps}_n \, S_z  \, \Pi^{\eps}_n\right] \right) + \Or(\eps^{n+1}) . \]
The first summand on the right-hand side of the above vanishes, since both operators in the commutator are bounded and of trace-per-unit-volume class (compare Ref.\ \onlinecite[Sec.\ 4.3]{MarcelliMonaco}). As for the second summand, we further rewrite
\[ - \eps \, \tau \left( \iu \left[\Pi^{\eps}_n \, X_2 \, \Pi^{\eps}_n, \Pi^{\eps}_n \, S_z  \, \Pi^{\eps}_n\right] \right) = - \eps \, \tau \left( \iu \left[X_2, \Pi^{\eps}_n \, S_z  \, \Pi^{\eps}_n\right] \right) + \eps \, \tau \left( \iu \left[(X_2)\su{OD}, \Pi^{\eps}_n \, S_z  \, \Pi^{\eps}_n\right] \right) \]
where $(X_2)\su{OD}$ refers now to the off-diagonal part of $X_2$ in the splitting of the Hilbert space induced by the orthogonal projection $\Pi^\eps_n$. 
The second summand in the above equality is once again the trace per unit volume of the commutator of two operatorators which are periodic, bounded, and trace-per-unit-volume class, and therefore vanishes. The first summand, instead, vanishes in view of Ref.\ \onlinecite[Prop.\ A.3(iv)]{MarcelliMonaco}. This concludes the proof.
\end{proof}

The above result generalizes Ref.\ \onlinecite[Corollary 5.5]{MaPaTe} where the vanishing of the spin torque in the NEASS is proved at first order in $\eps$.

By plugging the Taylor expansion \eqref{eqn: NEASS exp} for the NEASS in the above statement, we obtain in particular from Proposition \ref{prop:tau}\ref{it:tauAXex} that

\begin{corollary}
For any $n \in \N$, the trace per unit volume $\tau\left( \iu X_i \,[H_0, S_z] \, \Pi_n \right)$ does not depend on the choice of the origin, and in particular on the choice of the center of the cell $\FC_L$ in \eqref{eqn:defn FC1}.
\end{corollary}

Notice how, for $n=1$, the trace per unit volume mentioned in the above statement is exactly the one that computes $\sigma\sub{rot,12}\su{s}$, see \eqref{eqn:sigmaR}. In contrast, the conventional spin current operator $J\sub{1,conv}\su{s}$ has an expectation in the the NEASS $\Pi_n\su{\eps}$ whose definition as a trace per unit volume is never affected by these ambiguities, thanks to its periodicity. Therefore, we conclude from Theorem \ref{thm:notKubo} that the proper spin conductivity $\sigma\sub{prop,12}\su{s}$ is also well defined regardless of the choice of the origin.

%%%%% SECTION 3 %%%%%%%%%%%%%%%%%%%%%%%%%%%%%%
\section{Equality between conductance and conductivity for charge and spin transport}
\label{sec: ssigma=sG}
This section is devoted to an explanation of the relation between conductance and conductivity in two-dimensional models, as transport coefficients for both charge and spin. 

\subsection{Equality between Hall conductivity and Hall conductance}
\label{ssec: Hs=HG}
To properly introduce charge conductance in the mathematical framework we have adopted thus far, we resort to the notion of switching functions introduced before, but this time as functions of a spatial coordinate%
\footnote{In this section, we follow the conventions of Ref.\ \onlinecite{MarcelliPanatiTauber19}, so that whenever the configuration space $\X \simeq \Gamma \times \FC_1$ is discrete the coordinate $x_j \equiv n_j$ of a point $x \in \X$ refers to the $j$-th coordinate of the corresponding element of $\Gamma$: in other words, points in the same unit cell have the same spatial coordinates. This affects the definition of the switching functions that we are about to give, as well as of position operators $X_j$.} % 
$x_j$, $j \in \set{1,2}$. 

More precisely, consider a switching function $\lambda_j \colon \R \to [0,1]$ and fix $l_j > 0$. Set
\[ \Lambda_j(x_j) := \lambda_j \left( \frac{x_j}{2 \, l_j} - \frac{1}{2} \right). \]
Then, it is clear that $\Lambda_j$ is a smooth function satisfying
\[
\Lambda_j(x_j)=
\begin{cases}
0 & \text{if $x_j \le -l_j$,} \\
1 & \text{if $x_j \ge l_j$.}
\end{cases}
\]
With a slight abuse of notation, we denote by $\Lambda_j$ also the operator on $\Hi$ that multiplies by the function $\Lambda_j$. 

Charge conductance appears then as a linear response coefficient in the same setup introduced in Section \ref{sec:heuristic}. 
One considers an external perturbation $A = -\Lambda_2$, which can be interpreted as an electric potential of unit voltage drop over the interval $-l_2 \le x_2 \le l_2$ which then induces an electric field pointing in the positive 2-nd direction. Thus, the full perturbed time-dependent Hamiltonian \eqref{eqn: Heps(t)} reads now
\[
H^\eps(t):=H_0-\eps f(\eta t)\Lambda_2 ,\qquad \text{for } t\in \R,\, 0<\eta\ll 1.
\]
The unit voltage drop causes the generation of a charge current intensity in the $1$-st direction, which is represented by the operator $I_1\su{c}:=\iu [H_0,\Lambda_1]$. 
The linear response coefficient defined according to \eqref{eqn: deltaBA} is the \emph{transverse charge conductance}:
\begin{equation*}
%\label{eqn:G12c}
G_{12}\su{c}:=\Delta_{I_1\su{c},-\Lambda_2}.
\end{equation*}

A more tractable expression for $G_{12}\su{c}$ can be once again derived assuming the Kubo formula \eqref{eqn: DeltaAB Kubo}. Under the crucial hypothesis of the existence of a spectral gap around the Fermi energy, one can obtain from Ref.\ \onlinecite[Sect.3.(ii)]{Graf07} (see also Ref.\ \onlinecite[Chap. 1]{Marcelli18} which combines the previous reference with methods from Ref.\ \onlinecite{ElgartSchlein04})
\begin{equation}
\label{eqn: GDCF}
\Delta_{I_1\su{c},-\Lambda_2}\su{Kubo}=\iu\lim_{\eta\to0^+}\int_{-\infty}^{0}\di s\, \eu^{\eta s} \Tr\left(  I_1\su{c}\,\mathrm{e}^{\iu s H_0}[\Lambda_2,\Pi_0]\mathrm{e}^{-\iu s H_0}\right)=G\sub{Hall},
\end{equation}
with
\begin{equation}
\label{eqn:GHall}
G\sub{Hall}:=\iu \Tr(\Pi_0\left[[\Pi_0,\Lambda_1],[\Pi_0,\Lambda_2]\right]).
\end{equation}
It is worth mentioning that, in this setting, the LRA in the sense of \eqref{eqn: DeltaAB Kubo} has been put on a firm basis in Ref.\ \onlinecite{ElgartSchlein04} (see also Ref.\ \onlinecite{Marcelli21}), by coupling the time-adiabatic and zero-field limit setting the adiabatic parameter $\eta$ equal to the intensity $\eps$ of the perturbation. 

At this level we have obtained two response coefficients: a transverse \emph{conductivity}, which measures the response of an electric current to a perturbing field (and thus is the ratio of two intensive quantities), and a transverse \emph{conductance}, which is instead the current intensity per unit voltage drop (and thus is the ratio of two extensive quantities). From the analysis of their physical dimensions, one can observe that for a $d$-dimensional sample
\[ 
[G] = \frac{[j] \,\text{[length]}^{d-1}}{[E] \,\text{[length]}} = [\sigma] \, [\text{length}]^{d-2},
\]
and therefore it is evident that for $d=2$ their physical dimensions coincide. More can be proved: the two transport coefficients actually have the same value. This is obtained by comparing the two DCFs, namely \eqref{eqn: sigmahall} for $\sigma_{12}\su{c}$ and \eqref{eqn:GHall} for $G_{12}\su{c}$. Proofs of the identity
\begin{equation}
\label{eqn:DFCseq}
\iu \tau(\Pi_0\left[[\Pi_0,X_1],[\Pi_0,X_2]\right])=\iu \Tr(\Pi_0\left[[\Pi_0,\Lambda_1],[\Pi_0,\Lambda_2]\right]).
\end{equation}
are available in the literature, depending on specific details of the quantum model: In Ref.\ \onlinecite[Sect. 6]{AvronSeilerSimon} the above equality is proved for continuum models with covariant Fermi projections, associated with the Fermi energy in a spectral gap, while in Ref.\ \onlinecite[Lemma 8 \& Remark 3]{ElgartGrafSchenker} the same result is achieved in the discrete setting assuming only a mobility gap and without requiring any covariance structure.

\subsection{Spin conductance and spin conductivity: equality of Kubo-like terms}

At this point one might wonder if

\medskip

\noindent
\textbf{Question (Q$\mathbf{_{3}}$)} For spin transport, is it still possible to prove an analogous equality in two dimension between the respective quantities, namely is it possible to prove that spin conductivity $\sigma\su{s}_{12}$ agrees with spin conductance $G\su{s}_{12}$?

\medskip

\noindent
As was already stressed, the question becomes challenging in the spin-noncommuting case. A partial answer to Question \textbf{(Q$\mathbf{_{3}}$)} is given in Ref.\ \onlinecite{MarcelliPanatiTauber19}. 
In this work, the analysis is restricted to certain discrete models for crystalline solids: here, to illustrate the main ideas within it, we take a further simplification and adopt Assumption \ref{stronger} in the following.

\noindent
First we make a choice for the spin current operator to take into account, namely the proper spin current operator $J\sub{prop,1}\su{s}$ defined in \eqref{eqn: J_proper}. Similarly, the operator $I\sub{prop,1}\su{s}:=\iu[H_0,\Lambda_1 S_z]$ stands for the proper spin current intensity in the $1$-st direction. Hereinafter, since only operators of the ``proper'' guise are considered, we will omit the corresponding adjective from the notation.

As was explained in Section \ref{ssec: ssigma}, the expression for $\sigma_{12}\su{s}$ looks in general very different from a DCF. 
Analogous considerations can be done for the spin conductance. A full comparison of spin conductivity and spin conductance is complicated by the more cumbersome expressions for these transport coefficients. In Ref.\ \onlinecite{MarcelliPanatiTauber19} the authors were therefore led to single out the ``Kubo-like'' terms in their expressions, namely those which more resemble a DCF for the corresponding quantity. 
More precisely, we define
\begin{enumerate}[label={\rm (\alph*)},ref={\rm (\alph*)}]
\item \label{def:sigma_K} the \emph{Kubo-like spin conductivity} as 
\begin{equation*} 
%\label{Eq:sigma_K}
	{\sigma}_K\su{s}:=\tau (\Sigma_K\su{s}) \quad\text{ with }\quad \Sigma_K\su{s}:=\iu \Pi_0 \big[ [\Pi_0, X_1 S_z], [\Pi_0,X_2] \big] \Pi_0\,;
\end{equation*}
\item \label{def:G_K} the \emph{Kubo-like spin conductance}  as
\begin{align*}
%\label{Eq:G_K}
&G_K\su{s}(\Lambda_1,\Lambda_2):=\lim_{\substack{L\to\infty\\L\in 2\N+1}} \Tr(\,\chi_{1,L}\, \G_K\su{s}(\Lambda_1,\Lambda_2)\,\chi_{1,L}\,)\\[2mm]
& \text{ with } \chi_{1,L}:=\chi_{\{x\in \Gamma : \abs{x_1}\leq L/2\}},\quad \G_K\su{s}(\Lambda_1,\Lambda_2):=  \iu \Pi_0  \big[ [\Pi_0, \Lambda_1 S_z], [\Pi_0,\Lambda_2] \big]\Pi_0.
\nonumber
\end{align*}
\end{enumerate}

Notice a clear analogy with the corresponding DCFs for the Hall conductivity \eqref{eqn: sigmahall} and Hall conductance \eqref{eqn:GHall}, where, as was commented below the statement of Theorem \ref{thm: sigmahall}, there appear operators related to the observable ($X_1 S_z$ resp.\ $\Lambda_1 S_z$) and to the perturbation ($X_2$ resp.\ $\Lambda_2$).

By adapting ideas and techniques presented in Ref.\ \onlinecite[Sect. 6]{AvronSeilerSimon} from the continuum case (i.e.\ for covariant Schr\"odinger operators on the plane) to the discrete setting and to the spin-noncommuting case, and by using some results from Ref.\ \onlinecite{ElgartGrafSchenker} to establish trace class properties of the operators considered, we have that

\begin{theorem}
\label{thm:main2}
Let $H_0$ and $\Pi_0$ be as in Assumption~\ref{stronger}.
\begin{itemize}
\item[$-$]  The operator $\Sigma_K\su{s}$ satisfies
\begin{equation}
\label{eqn:Sigma notper: per+oddper}
{\left(\Sigma_K\su{s}\right)}_{\mathbf{m},\mathbf{n}}={\left(\Sigma_K\su{s}\right)}_{\mathbf{m-p},\mathbf{n-p}}- p_1 \, {\left(\mathcal{T}\sub{s}\right)}_{\mathbf{m-p},\mathbf{n-p}}\quad\text{for all $\mathbf{m},\mathbf{n},\mathbf{p}\in\Z^2$},
\end{equation}
where $\mathcal{T}\sub{s}$ is the \emph{spin torque-response operator}
\begin{equation}
\label{eqn:torqueresp}
\mathcal{T}\sub{s}:=\iu \Pi_0 \big[ [\Pi_0,S_z], [\Pi_0,X_2] \big]\Pi_0.
\end{equation}
Moreover,  ${\sigma}_K\su{s}$ is well-defined and satisfies
$$
{\sigma}_K\su{s}=\Tr(\chi_1\Sigma_K\su{s} \chi_1).
$$

\vspace{2mm}
\item[$-$]
Let $\Lambda_2$ be a fixed switching function in the $2$-nd direction. 
Assume that  $G_K\su{s}(\Lambda_1,\Lambda_2)$ is finite for at least one switching function $\Lambda_1$. 
Then $G_K\su{s}(\Lambda_1^{\prime} ,\Lambda_2)$ 
 is finite for any switching function $\Lambda_1^{\prime}$, and it is independent of the choice of $\Lambda_1^{\prime}$. 

\vspace{2mm}
\item[$-$] 
Finally, the equality
\begin{equation}
\label{eqn:KG12s=Ksigma12s}
\sigma_K\su{s} = G_K\su{s}( \Lambda_1, \Lambda_2)
\end{equation}
holds true. In particular, $G_K\su{s}$ is finite and independent of the choice of the switch functions $\Lambda_1,\Lambda_2$ in both directions.
\end{itemize}
\end{theorem}

We find interesting to stress a point from the proof of the above Theorem. In the argument, one shows that
\begin{align}
\label{eqn:KG=Ksigma}
G_K\su{s}( \Lambda_1, \Lambda_2)=\sigma_K\su{s}+\frac{1}{2}\lim_{\substack{L\to\infty \\ L\in 2\N+1}}\sum_{\substack{m_1\in\Z\\ \abs{m_1}\leq L/2}}\tau(\mathcal{T}\sub{s}),
\end{align}
where $\tau({\mathcal{T}\sub{s}})$ is the mesoscopic average of the spin torque-response operator defined in the statement. 
In the above, the second summand is a series of constant terms, which is either zero if $\tau(\mathcal{T}\sub{s})=0$, or divergent otherwise. 
Under the assumptions in the statement one can indeed show that $\tau(\mathcal{T}\sub{s})=0$.

\noindent
The terminology of ``spin torque-response'' adopted for $\tau({\mathcal{T}\sub{s}})$ is based on the observation that this average is nothing but the linear response coefficient $\Delta_{\iu [H_0,S_z],-X_2}\su{Kubo}$ of the spin torque $\iu [H_0,S_z]$ induced by a constant electric field in the second direction, according to the formula \eqref{eqn: DeltaAB Kubo}. Equivalently, its expression can be computed by the NEASS $\Pi_1^\eps$ at the first order, see Ref.\ \onlinecite[Proposition 5.4]{MaPaTe}. In view of our new result, Theorem \ref{thm: spin torque} above, the spin torque-response actually vanishes at all orders in $\eps$. 

\begin{acknowledgments}
This work was supported by the National Group of Mathematical Physics (GNFM--INdAM) within the project Progetto Giovani GNFM 2020. G.~M.\ gratefully acknowledges the financial support from the European Research Council (ERC), under the European Union's Horizon 2020 research and innovation programme (ERC Starting Grant MaMBoQ, no. 802901).
\end{acknowledgments}

\section*{Data Availability Statement}

Data sharing is not applicable to this article as no new data were created or analyzed in this study.

\nocite{*}
\bibliography{biblio}% Produces the bibliography via BibTeX.

\end{document}